\documentclass[letterpaper]{sig-alternate-2013}

\usepackage{comment}
\usepackage{color}
\usepackage[T1]{fontenc}
\usepackage[mathscr]{urwchancal} 
\usepackage{algpseudocode}
\usepackage{algorithm}
\usepackage{array}
\usepackage{mdwmath}
\usepackage{mdwtab}
\usepackage{eqparbox}
\usepackage{booktabs}

\usepackage[tight,normalsize,sf,SF]{subfigure}

\usepackage{stfloats}
\usepackage{url}

\usepackage{csquotes}

\usepackage{multirow}
\usepackage{epstopdf}

\newcommand{\matrx}[1]{\ensuremath\boldsymbol{\rm #1}}
\newcommand{\vect}[1]{\ensuremath\boldsymbol{\rm #1}}

\newcommand{\pnorm}[1]{\ensuremath{\|#1\|_p}}
\newcommand{\qnorm}[1]{\ensuremath{\|#1\|_q}}

\newcommand{\twonorm}[1]{\ensuremath{\|#1\|_2}}
\newcommand{\fnorm}[1]{\ensuremath{\|#1\|_F}}
\newcommand{\pqnorm}[1]{\ensuremath{\|#1\|_{p,q}}}

\newcommand{\onetwonorm}[1]{\ensuremath{\|#1\|_{1,2}}}

\newcommand{\topidx}[2]{\ensuremath{#1^{(#2)}}}
\newtheorem{theorem}{Theorem}

\newtheorem{remark}{Remark}

\newfont{\mycrnotice}{ptmr8t at 7pt}
\newfont{\myconfname}{ptmri8t at 7pt}
\let\crnotice\mycrnotice%
\let\confname\myconfname%
\permission{Permission to make digital or hard copies of all or part of this work for personal or classroom use is granted without fee provided that copies are not made or distributed for profit or commercial advantage and that copies bear this notice and the full citation on the first page. Copyrights for components of this work owned by others than ACM must be honored. Abstracting with credit is permitted. To copy otherwise, or republish, to post on servers or to redistribute to lists, requires prior specific permission and/or a fee. Request permissions from Permissions@acm.org.}
\conferenceinfo{KDD '15,}{August 11 - 14, 2015, Sydney, NSW, Australia}
\copyrightetc{\copyright~2015 ACM. \the\acmcopyr}
\crdata{ISBN 978-1-4503-3664-2/15/08\ ...\$15.00.\\
DOI: http://dx.doi.org/10.1145/2783258.2783403}

\begin{document}
\title{Longitudinal LASSO: Jointly Learning Features and Temporal Contingency for Outcome Prediction}
\numberofauthors{3}
\author{
\alignauthor
Tingyang Xu\\
       \affaddr{Department of Computer\\
        Science and Engineering}\\
       \affaddr{University of Connecticut}\\
       \affaddr{Storrs, CT, USA}\\
       \email{tix11001@engr.uconn.edu}
\alignauthor
Jiangwen Sun\\
       \affaddr{Department of Computer\\
        Science and Engineering}\\
       \affaddr{University of Connecticut}\\
       \affaddr{Storrs, CT, USA}\\       
       \email{javon@engr.uconn.edu}
\alignauthor
Jinbo Bi \titlenote{Correspondence should be adressed to Jinbo Bi.}\\
       \affaddr{Department of Computer\\
        Science and Engineering}\\
       \affaddr{University of Connecticut}\\
       \affaddr{Storrs, CT, USA}\\
       \email{jinbo@engr.uconn.edu}
}
\date{22 February 2015}

\maketitle

\begin{abstract}
Longitudinal analysis is important in many disciplines, such as the study of behavioral transitions in social science. Only very recently, feature selection has drawn adequate attention in the context of longitudinal modeling. Standard techniques, such as generalized estimating equations, have been modified to select features by imposing sparsity-inducing regularizers. However, they do not explicitly model how a dependent variable relies on features measured at proximal time points. Recent graphical Granger modeling can select features in lagged time points but ignores the temporal correlations within an individual's repeated measurements. We propose an approach to automatically and simultaneously determine both the relevant features and the relevant temporal points that impact the current outcome of the dependent variable. Meanwhile, the proposed model takes into account the non-{\em i.i.d} nature of the data by estimating the within-individual correlations. This approach decomposes model parameters into a summation of two components and imposes separate block-wise LASSO penalties to each component when building a linear model in terms of the past $\tau$ measurements of features. One component is used to select features whereas the other is used to select temporal contingent points. An accelerated gradient descent algorithm is developed to efficiently solve the related optimization problem with detailed convergence analysis and asymptotic analysis. Computational results on both synthetic and real world problems demonstrate the superior performance of the proposed approach over existing techniques.
\end{abstract}

%
%



\section{Introduction}
\label{introduction}
A longitudinal study collects and analyzes repeated measurements of a set of features for a group of subjects through time. Longitudinal analyses are important in many areas, such as in social and behavioral science \cite{Stappenbeck:2010,Fowler:2009,Bi:CSVM:2013}, in economics \cite{REEMtree,Arnold:2007:TCM}, in climate\cite{Lozano2009,Arnold:2007:TCM}, and in genetics \cite{Wang:2012}.  
For example, to predict binge drinking of college students, a longitudinal study may be designed to monitor them weekly or even daily in terms of multiple covariates, such as, the level of stress, status of negative affects and social behaviors \cite{Bi:CSVM:2013,DTCaffects2010Armeli}. The fluctuation of these covariates is used to analyze and predict binge drinking (the dependent or outcome variable) of a student at the current observation time point. Changes of the covariates in the proximal time points are anticipated to alter the likelihood that a student binge drinks at the current observation point. To precisely understand how covariates affect the outcome, the analysis has to model not only the current values of the covariates but also their proximal values as well as take into account the correlation structure in the repeated measurements. 

Typically, longitudinal data are analyzed by extending generalized linear models (GLM) with different assumptions, such as marginal models, random effects models, and transition models \cite{longitudinalAnalysis2002diggle}. For example, a marginal model regresses the outcome on the current observation of features but factors in a within-subject correlation matrix that is estimated for a few proximal time points. In contrast, a random effects model reflects the variability among individuals rather than the population average comparing with marginal models. For marginal modeling, generalized estimating equations (GEE) are the most widely used methods which estimate a predictive model to predict the current outcome together with correlations among different outcomes observed temporally. The resultant predictive models are generally more accurate than those of classic regression analysis that assumes independently and identically distributed (\textit{i.i.d.}) observations \cite{GEE:Liang:1986}. 
Research on feature selection in longitudinal data leads to a new family of methods based on the penalized GEE (PGEE)\cite{Fu:PEE:2003}. For random effects models, generalized linear mixture model(GLMM)\cite{Laird:1982:GLMM,McCulloch:2001:GLMM} is the major method. It explores natural heterogeneity across individuals in the regression coefficients and represents this heterogeneity by a probability distribution. 

None of those extensions of GLM aim to detect causal relationships from temporal changes of covariates to the outcomes of the current effect. In many studies, it is however necessary and insightful to model simultaneously the correlation among outcome records and the lagged causal effects of covariates \cite{DTCaffects2010Armeli}.
For example, psychologists have identified that there is lagged effect in the alcohol use behavior. An individual's drinking today may be a response to an elevated level of stress two days back rather than the current day. It is actually an important question for psychologists to find out both which temporal points and which covariates influence the current outcome the most. This lagged effect is not used by temporal marginal modeling to make predictions.

On the other hand, researchers have developed machine learning approaches for longitudinal analysis that predict an outcome using feature values at multiple time points \cite{Arnold:2007:TCM,Lozano2009}. For example, graphical Granger modeling \cite{Arnold:2007:TCM}, and grouped graphical Granger modeling\cite{Lozano2009} are insightful to explore the influences from past temporal information present in time series data in the modeling and understanding of the causal relationships.
These methods assume that past values of certain time series features causally affect an outcome variable, and hence construct a model based on these values to predict future outcomes. Often, they estimate causality relationship (causal graph) among all features.  However, these methods assume {\em i.i.d.} samples which are clearly violated in longitudinal data, and moreover they are incapable of selecting the most influential time points.

All existing methods either assume {\em i.i.d.} samples in Granger causality modeling or assume correlated samples but do not model {\em temporal} causal effects. Therefore, we propose a new learning formulation that constructs predictive models as functions of covariants not only from the current observation but also from multiple previous consecutive observations, and simultaneously determine the temporal contingency and the most influential features. The proposed method has the following advantages:
\begin{enumerate}
\item The proposed method makes predictions based on lagged data from current and previous time points. It decomposes the model coefficients into a summation of two components and impose different block-wise {\em least absolute shrinkage and selection operators} (LASSO) to the two components. One regularizer is used to detect the contingency of specific time points whereas the other is used to select covariates. 
\item The proposed method also learns simultaneously a structured correlation matrix from the data. The correlations among the outcomes themselves imply the changing trend of the outcomes in the proximal time points within each subject. 
\item We develop a family of methods where the outcome variable is assumed to follow a distribution from the exponential family, including Bernoulli, Gaussian and Poisson distributions. The formulations for these distributions are discussed in Section \ref{sec:Eg_family}.
\item We provide the convergence analysis in Section \ref{sec:convergence} and asymptotic analysis in Section \ref{sec:Asy_Ana} to show that the proposed algorithm can find the optimal solution for the predictive models.
\end{enumerate}
We have empirically compared the proposed method against the state of the art on both synthetic and real world datasets. The computational results demonstrate the effectiveness and the capability of our approach. 

\begin{figure}[h!]
\begin{center} 
\includegraphics[width=.3\textwidth]{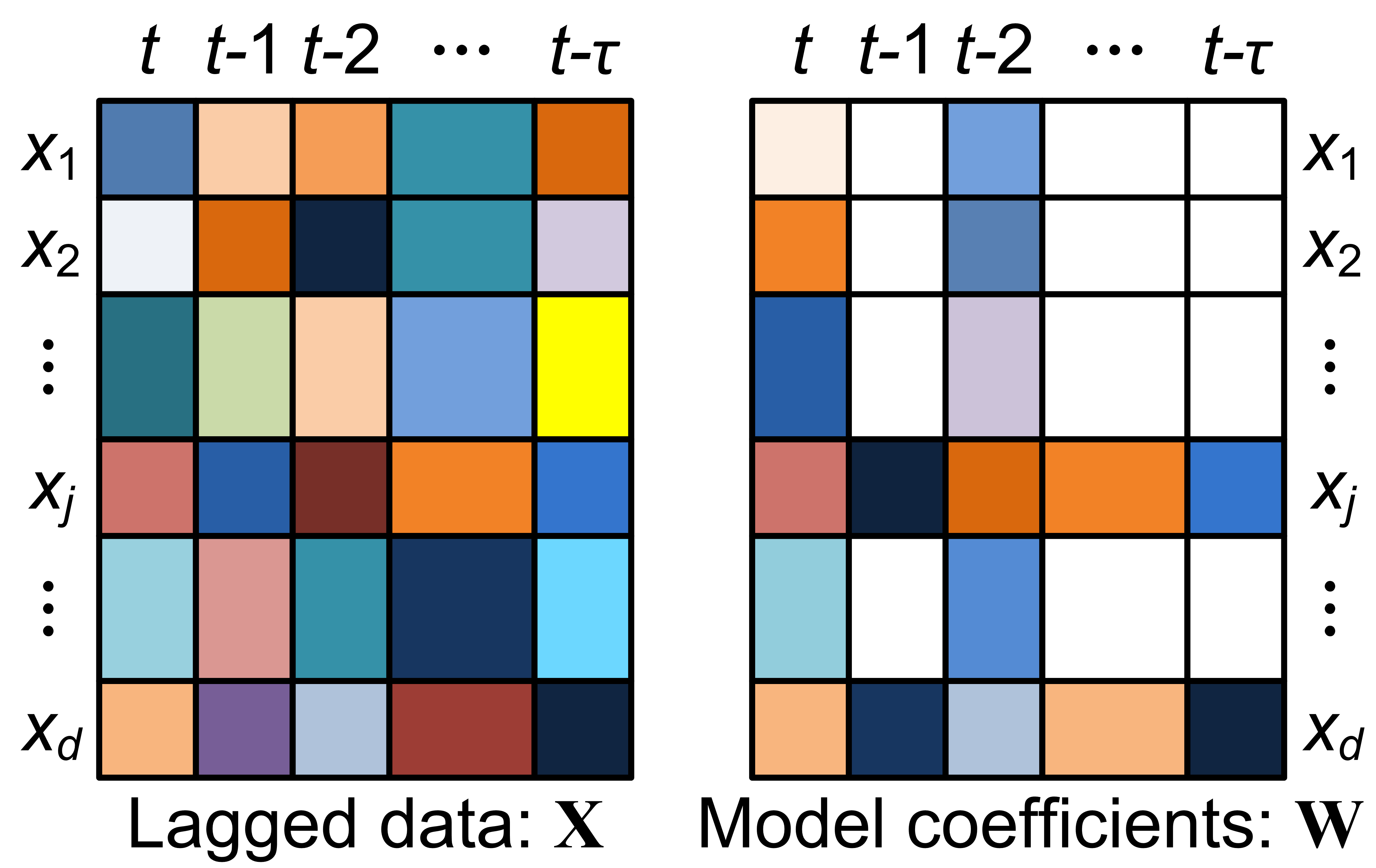}
\caption{The outcome $y_t$ at time $t$ can be relevant to multiple covariates $x_1, x_2, \cdots, x_d$ observed at current and several previous time points $t-1,t-2, \cdots, t-\tau$, which forms a data matrix $\matrx X$ (left). If we associate with each entry of this matrix a weight in our additive prediction model, then our model coefficients form a matrix $\matrx W$ (right). If the coefficient matrix is sparse, then the resultant model will be selective in terms of covariates and time points.} \label{fig:demo}
\end{center} \vspace{-0.3in}
\end{figure}

\section{Method}
\label{sec:formulation}
In our approach, the predictive model takes the form of the {\em trace} of the product of the lagged data $\matrx X$ and the model coefficient matrix $\matrx W$ as shown in Figure \ref{fig:demo}. The model coefficients are organized into a  matrix rather than a vector used in traditional analysis because this way reflects the structure in the lagged data. Note that the lagged observations of $y$ can also be included in the data matrix $\matrx X$ to be used in the predictive model. For notational convenience, we just use $\matrx X$ to represent the data that are used to form the model.

We first briefly review two most relevant sets of longitudinal analytics in Section \ref{sec:background} which will help elucidate the advantages of our proposed formulation.

\subsection{Preliminaries}\label{sec:background}
We introduce the notation that is used through out the paper. A bold lower case letter denotes a vector, such as $\vect v$. The $\pnorm{\vect v}$ refers to the $\ell_p$ norm of a vector $\vect v$, which is formed as $\pnorm{\vect v} = (\sum_{i=1}^d |v_i|^p)^{1/p}$, where $v_i$ is the $i$-th component of $\vect v$ and $d$ is the length of $\vect v$. A bold upper case letter denotes a matrix such as $\matrx{M}$. Similarly, $\vect m_{(i,)}$, $\vect m_{(,j)}$ and $m_{ij}$ represent the $i$-th row, $j$-th column and $(i,j)$-th component of $\matrx M$, respectively. The Frobenius norm and $\ell_{p,q}$ norm of a matrix $\matrx M$ refer, respectively, to $\|\matrx M \|_F$, which is equal to $(tr(\matrx M^\top \matrx M))^{1/2}$, and $\pqnorm{\matrx M}$, defined by $\left(\sum^n_{i=1}\left(\qnorm{\vect m_{(i,)}}\right)^p\right)^{1/p}$, where $n$ is the number of rows in $\matrx{M}$, and $tr(\matrx M)$ indicates the trace of $\matrx M$.  We assume that $\mbox{vect}(\matrx M)$ is the column-major vectorization of $\matrx M$, which is defined as $\mbox{vect}(\matrx M) = (\vect m_{(,1)}^\top, \cdots, \vect m_{(,k)}^\top)^\top$ assuming $k$ columns are in $\matrx M$. Then, $\langle\matrx M_1, \matrx M_2\rangle$ is the inner product of two matrices $\matrx M_1$ and $\matrx M_2$ that is computed as the inner product of $\mbox{vect}(\matrx M_1)$ and $\mbox{vect}(\matrx M_2)$. The operator $\mbox{reshape}(\vect v)$ re-shapes $\vect v$ into a matrix of a proper size determined by the specific context. 

Assume that we are given data of $m$ number of individuals on $d$ number of features (independent variables) that are repeatedly measured at $n_i$ time points for each individual $i$. The data of each individual $i$ is represented by a matrix $\matrx X^{(i)}$ of size ${d\times n_i}$, and $\vect x^{(i)}_t$ refers to the $d$-entry data vector of individual $i$ at time point $t$. Without loss of generality, we assume that all individuals have data at the same consecutive time points ($n_i = n$) to simplify the notation and the subsequent analysis. Data on the dependent variable (outcome) is also given in $\vect y^{(i)}$ of length $n$ that contains the observations at the $n$ time points for individual $i$. Typically, a longitudinal study aims to estimate the effect of covariates on the dependent variable.

\subsubsection{Granger Causality}
The notion of {\em Granger Causality} was introduced by the Nobel prize winning economist, Clive Granger, and has proven useful in time series analysis  \cite{granger1980testing}. It is based on the intuition that if a time series variable causally affects another, the past observations of the former should be useful in predicting the future outcome of the latter.

Specifically, a time series observation $x$ is said to {\em Granger cause} another time series outcome, $y$, if the regressing for $y$ in terms of past $y$ and $x$ is significantly better than the regressing just with past values of $y$. The so-called Granger test first performs two regressions:
\begin{equation} \label{eq:granger}
y^{(i)}_t=\sum_{j=1}^\tau  \left(a_j y^{(i)}_{t-j}+w_j^\top x^{(i)}_{t-j}\right),
\end{equation}
and $y^{(i)}_t=\sum_{j=1}^\tau  a_j y^{(i)}_{t-j}$, where $\tau$ is the maximum ``lag" in the past observations, and then uses a hypothesis test such as an F-test to determine if the outcome $y_t$ can be predicted significantly better from the past covariate $x$. Recent graphical Granger models \cite{Arnold:2007:TCM,Lozano2009} extend it from a single time series covariate $\vect x$ to multiple covariates $\matrx X$. They learn the coefficients $\vect a$ and $\vect w$'s with LASSO type of regularizers and evaluate if coefficients are non-zero for Granger causality.

\subsubsection{Generalized Estimating Equations (GEE)}
GEE estimates the parameters of a GLM while taking into account the correlations in the training examples. Similar to GLM, it assumes that the dependent variable comes from a class of distributions known as the exponential family. For each member in this family, there exists a link function that can be used to translate the nonlinear model into a linear model. The expectation of the outcome $y^{(i)}_t$ for subject $i$ at time $t$ is computed as:
\begin{equation}
\label{equ:exp_moments}
E(\topidx{y}{i}_t) = \mu^{(i)}_t = g^{-1}(\eta^{(i)}_t), 
\end{equation}
where $\mu^{(i)}_t$ represents the mean model, $g^{-1}$ is the inverse of a link function $g$ in a GLM \cite{McCu:Neld:1989:GLM},  and
$\topidx{\eta}{i}_t=\left(\topidx{\vect x}{i}_t\right)^\top \vect w$. The variance of $y^{(i)}_t$ is computed as $\mbox{var}(\topidx{y}{i}_t) =\mbox{var}(\mu^{(i)}_t)/\phi$ where $\phi$ is a scaling parameter that may be known or estimated. 

GEE presumes a so-called working correlation structure, typically denoted by $\matrx R(\vect\alpha)$, where $\vect \alpha$ is a parameter to be determined from data. The common choices of $\matrx R(\vect\alpha)$ include exchangeable, tri-diagonal and the first-order autoregressive (AR(1)) formula \cite{GEE:Liang:1986}. The exchangeable correlation structure, also called {\em equi-correlation}, assumes that  $corr(y_{it},y_{it'})=\alpha$ for all $t\neq t'$. The tri-diagonal structure uses a tridiagonal matrix as $\matrx R(\vect\alpha)$ where $corr(y_{it},y_{it'})= \alpha$ if $t' = t \pm 1$ or $0$ otherwise. 
The AR(1) formula assumes a correlation structure along continuous time, and uses $corr(y_{it},y_{it'})=\alpha^{|t-t'|}$.

To estimate the regression coefficients $\vect w$, GEE uses the the estimating equations that are formulated, in general, by setting the derivative of an appropriate loss function to 0. Although a loss function may not be explicitly written out, the estimating equations always can be computed by 
\begin{equation}
\label{equ:GEE}
EE(\vect w,\vect \alpha) = \sum^m_{i=1}\left(\matrx D^{(i)}\right)^\top\left(\matrx \Sigma^{(i)}\right)^{-1}\vect s^{(i)}=0.
\end{equation}
where the $n \times d$ matrix $\matrx D^{(i)}=\partial\vect\mu^{(i)}/\partial\vect w$ where $\vect\mu^{(i)}$ combines all $\mu_t^{(i)}, \forall t=1,\cdots,n$ into a vector, $\vect s^{(i)} = \vect y^{(i)} - \vect\mu^{(i)}(\vect w)$. The $n \times n$ matrix $\matrx \Sigma^{(i)}$ is the estimated covariance structure as:
\begin{equation}
\label{equ:sigma}
\matrx \Sigma^{(i)}(\vect{\alpha}) = \left(\matrx A^{(i)}\right)^{1/2}\matrx R(\vect \alpha)\left(\matrx A^{(i)}\right)^{1/2}/\phi
\end{equation}
where $\matrx A^{(i)}$ is an $n \times n$ diagonal matrix with $\mbox{var}(\mu^{(i)}_t)$ as the $t$-th diagonal element. Algorithms are given in \cite{GEE:Liang:1986}  to compute $\vect w$ and $\vect \alpha$ for the different choices of $\matrx R(\vect\alpha)$.

\subsection{The Proposed Formulation}
\label{subsec:formulation_LR}
In our approach, each training example consists of the current and $\tau$ previous records of the repeated measurements. Let
\begin{equation*}
\matrx X_{(i;t)} = [\vect x^{(i)}_{t}, \vect x^{(i)}_{t-1}, \cdots, \vect x^{(i)}_{t-\tau}] \end{equation*}
be a ${d\times(\tau + 1)}$ data matrix for subject $i$. Given $T$ total measurements for each subject, the index $t$ of $\matrx X_{(i;t)}$ starts from $\tau+1$ in order to have enough previous observations in the first training example. Hence, there are totally $n=T-\tau$ training examples for each subject. If $\matrx X_{(i;t)}$ includes previous $\tau+1$ values of $y^{(i)}$ as a feature, then the model $y_t^{(i)} = tr\left(\matrx X^\top_{(i;t)}\matrx W\right)$ where $\matrx W = [\vect w_0, \vect w_1, \cdots, \vect w_\tau]$ essentially gives the same model like Eq.(\ref{eq:granger}) in the graphical Granger models.

The Granger models would assume that the training examples are {\em i.i.d.}. However, the consecutive examples are not mutually independent because they contain overlapping records (e.g., $\matrx X_{(i;t)}$ and $\matrx X_{(i;t+1)}$ share $\tau-1$ records $\vect x^{(i)}_t$, $\cdots$, $\vect x^{(i)}_{t-\tau+1}$). GEE provides a mechanism to estimate the sample correlation simultaneously while constructing predictive models, and to extend the linear models to generalized linear models. To apply GEE to our model, we replace $\eta_t^{(i)}$ used in GEE by the following formula
\begin{equation} \label{equ:AVR}
\eta^{(i)}_t=tr\left(\matrx X^\top_{(i;t)}\matrx W\right).
\end{equation}
Substituting Eq.(\ref{equ:AVR}) for $\eta$ in Eq.(\ref{equ:exp_moments}) yields a formulation similar to GEE. The regression coefficients $\matrx W$ can be estimated through the well-developed GEE estimators. In particular, the quasi-likelihood methods of GEE estimate $\matrx W$ by minimizing a loss function that is defined via the model deviance. The model deviance measures the difference between the log-likelihood of the estimated mean model $\vect \mu^{(i)}$ and that of the observed values $\vect y^{(i)}$. 
For instance, the model deviance for a linearly regressive response is written by $Dev^{(i)}(\matrx W,\vect \alpha)=(\vect y^{(i)}-\vect\mu^{(i)})^\top\matrx R(\vect \alpha)(\vect y^{(i)}-\vect\mu^{(i)})$ where $\vect y^{(i)}$ contains the observed responses for subject $i$, and $\vect\mu^{(i)}$ is the estimated expectations of $y$ for subject $i$. If the response follows an arbitrary distribution, the model deviance may not correspond to an explicit function. For the exponential family, it takes a special form as discussed in Theorem 1 below, which is still complicated. We denote by $Dev^{(i)}(\matrx W,\vect \alpha)$ the deviance occurred on subject $i$. GEE minimizes a loss function of $\sum_{i=1}^m Dev^{(i)}(\matrx W,\vect \alpha)$ for the optimal $\matrx W$ by solving the {\em estimating equations}, i.e., taking the derivatives of the loss function and setting them to $0$.

Now, to select among features and discover the most influential time points in predicting $y$ over time, (and also to control the model capacity,) we apply regularizers to the model parameters. We first decompose $\matrx W$ into a summation of two components as $\matrx W = \matrx U + \matrx V$ and apply different regularizers to $\matrx U$ and $\matrx V$. The block-wise LASSO, such as the $\ell_{1,2}$ matrix norm, is widely-used in multi-task learning or feature selection with group structures, but has not been explored within the GEE setting. To the best of our knowledge, it has not been studied in longitudinal analytics how to produce shrinkage effects simultaneously on both features and contingent temporal records through proper regularization. The general $\ell_{1,p}$ matrix norm  \cite{zhang2010probabilistic} calculates the sum of the $\ell_p$ norms of the rows in a matrix. Regularizers based on the $\ell_{1,p}$ norms encourage row sparsity by shrinking the entire rows to have zero entries. 

In our parameter matrix $\matrx W$, rows correspond to features and columns correspond to the observation time points. If we apply the $\ell_{1,2}$ norm to $\matrx U$ (row-wisely), the optimal solution of ${\matrx U}$ will contain rows with all zero entries. Thus, a selected subset of features in the $\tau+1$ observations will be used in the predictive model to predict the current outcome. The $\ell_{1,2}$ norm of $\matrx V^\top$ (column-wisely) encourages to select among columns of $\matrx V$. If the $k$-th column of $\matrx V$ contains the largest values in the selected columns,  the current outcome is most contingent on the previous $(k-1)$-th record, thus having the $(k-1)$ ``lagged" effect. 
Overall, we solve the following optimization problem for the best model parameters $\matrx W$ which is computed as $\matrx U + \matrx V$:
\begin{eqnarray}
\label{equ:dev_penality}
\min_{\matrx U, \matrx V}~~~\sum^m_{i=1} Dev^{(i)}(\matrx U + \matrx V, \vect \alpha)+ \lambda_1\onetwonorm{\matrx U} + \lambda_2 \onetwonorm{\matrx V^\top}
\end{eqnarray} 
where $\matrx W$ in the deviance is simply replaced by $\matrx U + \matrx V$. 

The optimization of Eq.(\ref{equ:dev_penality}) is challenging. In general, even solving the GEE formulation is not easy as it estimates not only the model expectation but also the variance term $\matrx \Sigma^{(i)}$. The algorithm that solves the GEE (i.e., the estimating equations) applies the Newton-Raphson method in the iterative reweighted least squares (IRLS) procedure \cite{Fu:PEE:2003} to estimate $\vect  w$ and $\matrx \Sigma^{(i)}$. However, this method does not solve any formula that uses regularizers. By modifying the Newton-Raphson method or shooting algorithm \cite{Fu:PEE:2003}, it can be extended only to the regularizers that are decomposable into individual parameters $w_j$. For instance, the $\ell_1$ vector norm of $\vect w$ can be decomposed into the summation of individual $|w_j|$, $j=1,\cdots, d$. The $\ell_{1,2}$ matrix norm, unfortunately, can not be decomposed in such a way. Therefore, we have developed an accelerated gradient descent method based on the fast iterative shrinkage-thresholding algorithm (FISTA) \cite{Beck:2009:FISTA}. Further, the following theorem shows that Eq.(\ref{equ:dev_penality}) is a convex optimization problem in terms of $\matrx W$. Our algorithm can be proved to find the global optimal solution $\matrx W$ of Eq.(\ref{equ:dev_penality}) when $\vect \alpha$ is fixed (to a consistent estimate given by GEE).

\begin{theorem}
\label{thm:convex}
The first term of Eq.(\ref{equ:dev_penality}) is convex and continuously differentiable with respect to $\matrx U$ and $\matrx V$ if the distribution of $\vect y^{(i)}$ is in a natural exponential family and the link function is continuous.
\end{theorem}
 
\begin{proof}
First, let us recall that the probability density function of a distribution in the exponential family takes the following form:
\begin{equation*}
\label{equ:expfmly}
f(y^{(i)}_t)=\exp\left\{\frac{y^{(i)}_t\eta^{(i)}_t-b(\eta^{(i)}_t)}{a^{(i)}_t(\phi)}+c(y^{(i)}_t,\phi)\right\},
\end{equation*}
where $a^{(i)}_t(\phi)$, $b(\eta^{(i)}_t)$, and $c(y^{(i)}_t,\phi)$ are known functions and specified for each member of the exponential family, and $\eta^{(i)}_t$ is a parameter in the mean as defined in Eq.(\ref{equ:exp_moments}). Typically, $a^{(i)}_t(\phi)=\phi$. Then, the deviance of the exponential family can be computed as
\begin{equation*}
\label{equ:expfamily}
Dev=2\frac{\sum_{i=1}^{m}\left(y^{(i)}_t(\tilde{\eta}^{(i)}_t-\hat{\eta}^{(i)}_t)-b(\tilde{\eta}^{(i)}_t)+b(\hat{\eta}^{(i)}_t)\right)}{\phi},
\end{equation*}
where $\tilde{\eta}^{(i)}_t$ denotes the true value under a saturated model, $\hat{\eta}^{(i)}_t$ denotes the fitted values of the model. Thus, $\tilde{\eta}^{(i)}_t$ and $b(\tilde{\eta}^{(i)}_t)$ are constant in model fitting. The derivative of $b$ always satisfies $b'({\eta}^{(i)}_t)={\mu}^{(i)}_t$. Moreover, it has been proved that $b(\hat{\eta}^{(i)}_t)$ is a convex function on the natural parameter space $\matrx H = \{\vect{\hat{\eta}}|b(\vect{\hat{\eta}})<\infty\}$ \cite{expFamily2005Severini}. Thus, the deviance contains either linear terms or a convex term with respect to $\hat{\eta}$. In our model (\ref{equ:AVR}), $\hat{\eta}$ is linear with respect to $\matrx W$. Hence, the deviance term in Eq.(\ref{equ:dev_penality}) is convex with respect to $\matrx U$ and $\matrx V$.
 
Moreover, it is true that $b'(\hat{\eta}^{(i)}_t)=\hat{\mu}^{(i)}_t=g^{-1}(\hat{\eta}^{(i)}_t)$ which is the inverse of a continuous link function \cite{expFamily2005Severini}.  The first term of Eq.(\ref{equ:dev_penality}) is continuously differentiable with respect to $\matrx U$ and $\matrx V$. Thus, theorem \ref{thm:convex} holds.
\end{proof}

\subsection{Optimization Algorithm}
\label{subsec:optimization}
To solve Eq.(\ref{equ:dev_penality}), we design an alternating optimization algorithm that alternates between optimizing two working sets of variables: one set consisting of $\matrx U$ and $\matrx V$ and the other consisting of $\vect \alpha$. 

\smallskip
\noindent\textit{\textbf{(a) Find $\matrx U$ and $\matrx V$ when $\vect \alpha$ is fixed}}

When $\vect \alpha$ is fixed, the objective function of Eq.(\ref{equ:dev_penality}), denoted by $f(\matrx U,\matrx V)$, is convex with a continuously differentiable part $\ell(\matrx U, \matrx V)$ that is the deviance and a nonsmooth part $R(\matrx U,\matrx V)$ that constitutes the two regularizers. 
We hence have
\begin{equation*}
f(\matrx U, \matrx V) = \ell(\matrx U, \matrx V) + R(\matrx U, \matrx V).
\end{equation*}
We develop a FISTA algorithm in the following iterative procedure to find optimal $\matrx U$ and $\matrx V$.

Denote the iterates at the $k$-th iteration by $\matrx U_k$ and $\matrx V_k$. Let $\nabla_{\matrx U}\ell(\matrx U, \matrx V)$, $\nabla_{\matrx V}\ell(\matrx U, \matrx V)$ be the partial derivative of $\ell(\matrx U, \matrx V)$ with respect to $\matrx U$ and $\matrx V$, respectively, 
For any given point $(\tilde{\matrx U}, \tilde{\matrx V})$, the following $Q_{L, \tilde{\matrx U}, \tilde{\matrx V}}(\matrx U, \matrx V)$ is a {\em well-defined} proximal map for the non-smooth $R$
\begin{equation*}
\begin{split}
Q_{L, \tilde{\matrx U}, \tilde{\matrx V}}(\matrx U, \matrx V) & = \ell(\tilde{\matrx U}, \tilde{\matrx V}) + R(\matrx U, \matrx V) \\
&+ \langle\nabla_{\matrx U}\ell(\tilde{\matrx U}, \tilde{\matrx V}),\matrx U - \tilde{\matrx U}\rangle + \frac{L}{2}\fnorm{\matrx U - \tilde{\matrx U}}^2 \\
&+ \langle\nabla_{\matrx V}\ell(\tilde{\matrx U}, \tilde{\matrx V}),\matrx V - \tilde{\matrx V}\rangle + \frac{L}{2}\fnorm{\matrx V - \tilde{\matrx V}}^2.
\end{split}
\end{equation*}
If $\ell(\matrx U, \matrx V)$ has Lipschitz continuous gradient with Lipschitz modulus $L$. Then, according to the Lemma 2.1 in \cite{Beck:2009:FISTA}, the inequality 
\begin{equation*}
\label{equ:Lipschitz}
f(\matrx U, \matrx V) \le Q_{L,\tilde{\matrx U}, \tilde{\matrx V}}(\matrx U, \matrx V).
\end{equation*}
holds indicating that $Q_{L,\tilde{\matrx U}, \tilde{\matrx V}}(\matrx U, \matrx V)$ is the upper
bound of $f(\matrx U, \matrx V)$.

Starting from an initial point $(\matrx U_0, \matrx V_0)$, we iteratively search for the optimal solution. At each iteration $k$, we first use the iterates $(\matrx U_{k-1},\matrx V_{k-1})$ and $(\matrx U_{k-2},\matrx V_{k-2})$  to compute (at the first iteration, $(\tilde{\matrx U}_{1},\tilde{\matrx V}_{1})=(\matrx U_{0},\matrx V_{0})$)
\begin{equation}\label{equ:interm_UV}
\begin{split}
\tilde{\matrx U}_{k} = \matrx U_{k-1} + \left(\frac{t_{k-1} -1}{t_k}\right)(\matrx U_{k-1} - \matrx U_{k-2}),\\
\tilde{\matrx V}_{k} = \matrx V_{k-1} + \left(\frac{t_{k-1} -1}{t_k}\right)(\matrx V_{k-1} - \matrx V_{k-2}),
\end{split}
\end{equation}  
where $t_k$ is a scalar and updated at each iteration as:
\begin{equation}
\label{equ:update_t}
t_{k+1} = \frac{1+\sqrt{1+4t_k^2}}{2}.
\end{equation}
Then, we solve the following problem
\begin{equation}
\label{equ:update_UV}
\begin{split}
\min_{\matrx U, \matrx V} ~~~~&\langle\nabla_{\matrx U}\ell_k,\matrx U - \tilde{\matrx U}_{k}\rangle + \frac{L}{2}\fnorm{\matrx U - \tilde{\matrx U}_{k}}^2 \\
&+ \langle\nabla_{\matrx V}\ell_k,\matrx V - \tilde{\matrx V}_{k}\rangle + \frac{L}{2}\fnorm{\matrx V - \tilde{\matrx V}_{k}}^2 \\
&+ R(\matrx U, \matrx V)
\end{split}
\end{equation} 
for a solution $(\matrx U_k, \matrx V_k)$, where $\nabla_{\matrx U}\ell_k$ and $\nabla_{\matrx V}\ell_k$ are respectively the partial derivatives of $\ell$ computed at $(\tilde{\matrx U}_{k}, \tilde{\matrx V}_{k})$, and $L$ acts as a learning step size.

Since there is no interacting term between $\matrx U$ and $\matrx{V}$ in Eq.(\ref{equ:update_UV}), the problem can be decomposed into two separate subproblems as follows:
\begin{equation}
\label{equ:solve_U}
\min_{\matrx U} \langle\nabla_{\matrx U}\ell_k,\matrx U - \tilde{\matrx U}_{k}\rangle + \frac{L}{2}\fnorm{\matrx U - \tilde{\matrx U}_{k}}^2 + \lambda_{1} \onetwonorm{\matrx U}, 
\end{equation} 
\begin{equation}
\label{equ_solve_V}
\min_{\matrx V} \langle\nabla_{\matrx V}\ell_k,\matrx V - \tilde{\matrx V}_{k}\rangle + \frac{L}{2}\fnorm{\matrx V - \tilde{\matrx V}_{k}}^2 + \lambda_{2} \onetwonorm{\matrx V^\top}.
\end{equation} 
The two subproblems share the same structure and thus can be solved following the same procedure. Hence, we only show how to solve (\ref{equ:solve_U}) for the best $\matrx U$.

Eq.(\ref{equ:solve_U}) is equivalent to the following problem
\begin{equation*}
\min_{\matrx U} \frac{1}{2}\left\lVert\matrx U - \left(\tilde{\matrx U}_{k} - \frac{1}{L}\nabla_{\matrx U}\ell_k\right)\right\rVert_F^2 + \frac{\lambda_{1}}{L} \onetwonorm{\matrx U}
\end{equation*}
after omitting constants, and this problem has a closed-form solution where each row of $\matrx U_k$, $\matrx U^k_{(i,)}$ is:
\begin{equation*}
\matrx U^{k}_{(i,)} = \max\left(0, 1-\frac{\lambda_1}{L\twonorm{\topidx{\matrx P}{k}_{(i,)}}}\right)\topidx{\matrx P}{k}_{(i,)},
\end{equation*}
and $
\topidx{\matrx P}{k} = \tilde{\matrx U}_{k} - \frac{1}{L}\nabla_{\matrx U}\ell_k$.
The gradient vector $\nabla_{\matrx U}\ell_k$ (i.e., the gradient of the deviance) can be computed by Eq.(\ref{equ:GEE}) with the fixed $\vect \alpha$, i.e.
\begin{equation}
\label{equ:gradient}
\nabla_{\matrx U}\ell_k =  \mbox{reshape}\left(\sum^m_{i=1}\left(\matrx D^{(i)}\right)^\top\left(\matrx \Sigma^{(i)}\right)^{-1}\vect s^{(i)}_k\right)
\end{equation}
where $\vect s^{(i)}_k=\vect y^{(i)}-\vect \mu^{(i)}$, and $\mu^{(i)}_t = g^{-1}(tr(\matrx X^\top_{(i;t)}(\tilde{\matrx U}_k + \tilde{\matrx V}_k)))$.

In the above discussion, the Lipschitz modulus $L$ is computed and given. However, the calculation of $L$ can be computational expensive. We therefore follow the similar argument in \cite{MTL:2012:Gong} to find a proper approximation $L_k$ at each iteration $k$ starting from $L_0 > 0$. Recall that the Lipschitz constant $L$ is defined:
\begin{equation*}
L=\max_{\matrx W}\lambda_{\max}\left(\nabla\nabla \ell_{\matrx W}\right)
\end{equation*}
where $\lambda_{\max}(\cdot)$ indicates the maximum singular value of the Hessian of $\ell$. Decompose the Hessian matrix $\left.\nabla\nabla \ell_{\matrx W}\right|_{\matrx W \rightarrow 0}$ into $\matrx M ^\top \matrx M$ where $\matrx M\in \mathbb{R}^{d(\tau+1) \times q}$ and $q$ is the rank of the Hessian matrix. We have an upper bound of $L$ as follows:
\begin{equation}\label{eq:upperL}
L\le||\matrx M||_{\infty,1}||\matrx M^\top||_{\infty,1}.
\end{equation}

We use the upper bound $\tilde{L}$ in Eq.(\ref{eq:upperL}) as $L$ in our iterations. Using this upper bound may increase the number of iterative steps for convergence.
Algorithm \ref{alg:solve_uv} summarizes the steps for finding optimal $\matrx U$ and $\matrx V$ with fixed $\vect \alpha$.
\begin{algorithm}[h]
   \caption{ Search for optimal $\matrx U$ and $\matrx V$ with fixed $\vect \alpha$}
   \label{alg:solve_uv}
\begin{algorithmic}
   \State {\bfseries Input:} $\matrx X$, $\vect y$, $\matrx \Sigma$, $\lambda_1$, $\lambda_2$
   \State {\bfseries Output:} $\matrx U$, $\matrx V$
   \State 1. $k$ = 1, compute $\tilde{L}$ and initialize $t_1=1$, $\matrx U_0= \tilde{\matrx U}_1 = \vect 0$ and $\matrx V_0= \tilde{\matrx V}_1 = \vect 0$;
   \State 2. Solve Eq.(\ref{equ:update_UV}) to obtain $\matrx U_{k}$ and $\matrx V_{k}$.  
   \State 3. Compute $t_{k+1}$ by Eq.(\ref{equ:update_t}).
   \State 4. Compute $\tilde{\matrx U}_{k+1}$ and $\tilde{\matrx V}_{k+1}$ by Eq.(\ref{equ:interm_UV}).
   \State 5. $k=k+1$.
   \State Repeat $2\sim 5$ until convergence.
\end{algorithmic}
\end{algorithm}

\smallskip
\noindent\textit{\textbf{(b) Find $\vect \alpha$ when $\matrx U$ and $\matrx V$}} are fixed

When $\matrx U$ and $\matrx V$ are fixed, the regularizers no longer appear in the objective of Eq.(\ref{equ:dev_penality}). Eq.(\ref{equ:dev_penality}) is degenerated into just the GEE formula with $\vect \alpha$ as the variables. Hence, ${\vect\alpha}$ can be estimated via the standard GEE procedure, i.e., from the current Pearson residuals defined by:
\begin{equation*}
\topidx{\gamma}{i}_t = \frac{\topidx{y}{i}_t-tr\left(\left(\matrx X_{(i;t)}\right)^\top (\matrx U + \matrx V)\right)}{(\topidx{\sigma}{i}_{t,t})^{(1/2)}}.
\end{equation*}
where $\sigma^{(i)}_{t,t}$ is the $t$-th diagonal entry in the matrix $\matrx \Sigma^{(i)}$ \cite{GEE:Liang:1986}. The specific estimator of ${\vect\alpha}$ depends on the choices of $\matrx R(\vect\alpha)$. This GEE-based procedure has been shown to find a {\em consistent} estimate of $\vect \alpha$ \cite{GEE:Liang:1986}.

Let $N=mn$ be the total number of training examples, and $p=d(\tau + 1)$ be the practical number of parameters in $\matrx W$. A general approach to estimating $\matrx R$ is given by:
\begin{equation} \label{eq:compute_R}
r_{j,k} = \sum_{i=1}^m \frac{\topidx{\gamma}{i}_j \topidx{\gamma}{i}_k}{N-p},
\end{equation}
for $j=1,\cdots,n$, and $k=1,\cdots, n$. In addition, the scaler parameter $\phi$ in Eq.(\ref{equ:sigma}) can be estimated as follows:
\begin{equation} \label{eq:phi}
{\phi}=(N-p)/\sum_{i=1}^m \sum_{t=1}^n \left(\topidx{\gamma}{i}_t\right)^2.
\end{equation}

Algorithm \ref{alg:overall} depicts the overall procedure for solving Eq.(\ref{equ:dev_penality}).
\begin{algorithm}[h] 
   \caption{ Main algorithm - Jointly select features and temporal points}
   \label{alg:overall}
\begin{algorithmic}
   \State {\bfseries Input:} $\matrx X$, $\vect y$, $\lambda_1$, $\lambda_2$
   \State {\bfseries Output:} $\matrx U$, $\matrx V$
   \State 1. Set $\matrx R(\alpha)$ = $\matrx I$;
   \State 2. Solve for $\matrx U$ and $\matrx V$ using Algorithm \ref{alg:solve_uv}.
   \State 3. Estimate $\alpha$ using a proper estimator in \cite{GEE:Liang:1986} and compute $\matrx R(\alpha)$ by Eq.(\ref{eq:compute_R}) and $\phi$ by Eq.(\ref{eq:phi}).
   \State Repeat $2\sim 3$ until convergence.
\end{algorithmic}
\end{algorithm}

\section{Theoretical Analysis}\label{sec:Con_Ana}
We provide a convergence analysis for Algorithm 1 and an asymptotic analysis for the proposed formulation. 

\subsection{Convergence Analysis} \label{sec:convergence}
We show that Algorithm 1 converges to the optimal solution with a convergence rate of $O(1/k^2)$. The proof follows largely the arguments in \cite{Beck:2009:FISTA}. We only provide a sketch here.

\begin{theorem}
Let {$\matrx U_k$} and {$\matrx V_k$} be the pair of the matrix generated by Algorithm \ref{alg:solve_uv}. Then for any $k\ge 1$
\begin{equation*}
f(\matrx U_k,\matrx V_k)-f(\hat{\matrx U},\hat{\matrx V})\le \frac{2\tilde{L}\left(||\matrx U_0-\hat{\matrx U}||^2_F+||\matrx V_0-\hat{\matrx V}||^2_F\right)}{(k+1)^2}
\end{equation*}
where $(\hat{\matrx U}, \hat{\matrx V})$ is a globally optimal solution of Eq.(\ref{equ:dev_penality}).
\end{theorem}
\begin{proof}
We start with defining the following quantities
\begin{align*}
v_k=&f(\matrx U_k,\matrx V_k)-f(\hat{\matrx U},\hat{\matrx V}),\\
a_k=&\frac{2}{L_k}t^2_kv_k,\\
b_k=&||t_k\matrx U_k-(t_k-1)\matrx U_{k-1}-\hat{\matrx U}||^2_F\\
    +&||t_k\matrx V_k-(t_k-1)\matrx V_{k-1}-\hat{\matrx V}||^2_F, \\
c  =&||\tilde{\matrx U}_1-\hat{\matrx U}||^2_F+||\tilde{\matrx V}_1-\hat{\matrx V}||^2_F\\
   =&||\matrx U_0-\hat{\matrx U}||^2_F+||\matrx V_0-\hat{\matrx V}||^2_F, 
\end{align*}
where $\tilde{\matrx U}_1=\matrx U_0$, $\tilde{\matrx V}_1=\matrx V_0$, and subsequent $\tilde{\matrx U}_k$ and $\tilde{\matrx V}_k$ are defined by Eq.(\ref{equ:interm_UV}).
Following the proof of Theorem 4.4 in \cite{Beck:2009:FISTA}, in the first iteration, given $t_1=1$, we have $a_1=\frac{2}{L_1}v_1$, and $b_1=||\matrx U_1-\hat{\matrx U}||^2_F-||{\matrx V}_1-\hat{\matrx V}||^2_F$. We show that $a_1 + b_1 \le c$ by applying Lemma 2.3 in \cite{Beck:2009:FISTA}, which yields
\begin{align*}
&f(\hat{\matrx U},\hat{\matrx V})-f(\matrx U_1,\matrx V_1)=-v_1\\
\ge& \frac{L_1}{2}||\matrx U_1-\tilde{\matrx U}_1||^2_F+L_1\langle\tilde{\matrx U}_1-\hat{\matrx U},\matrx U_1-\tilde{\matrx U}_1\rangle\\
&+\frac{L_1}{2}||\matrx V_1-\tilde{\matrx V}_1||^2_F+L_1\langle\tilde{\matrx V}_1-\hat{\matrx V},\matrx V_1-\tilde{\matrx V}_1\rangle\\
=&\frac{L_1}{2}(||\matrx U_1-\hat{\matrx U}||^2_F-||\tilde{\matrx U}_1-\hat{\matrx U}||^2_F)\\
&+\frac{L_1}{2}(||\matrx V_1-\hat{\matrx V}||^2_F-||\tilde{\matrx V}_1-\hat{\matrx V}||^2_F).
\end{align*}
Reorganizing the above inequality yields \begin{align*} \frac{2}{L_1}t_1^2 v_1+||\matrx U_1&-\hat{\matrx U}||^2_F+||\matrx V_1-\hat{\matrx V}||^2_F\le\\
& ||\tilde{\matrx U}_1-\hat{\matrx U}||^2_F+||\tilde{\matrx V}_1-\hat{\matrx V}||^2_F
\end{align*}
Thus, $a_1+b_1\le c$ holds. 
 
Then, according to Lemma 4.1 in \cite{Beck:2009:FISTA}, we have for every $k\ge 1$, $a_k-a_{k+1}\ge b_{k+1}-b_k$, together with $a_1+b_1\le c$, which derives into the following inequality,
$$c\ge a_1+b_1\ge a_2+b_2\ge \dots\ge a_k+b_k\ge a_k.$$
Therefore, we obtain that
\begin{equation}\label{eq:proof}
\frac{2}{L_k}t^2_kv_k\le ||\matrx U_0-\hat{\matrx U}||^2_F+||\matrx V_0-\hat{\matrx V}||^2_F,
\end{equation}
Given $t_k$ is updated according to Eq.(\ref{equ:update_t}), it is easy to show that $t_k\ge \dfrac{(k+1)}{2}$. Substituting this inequality into Eq.(\ref{eq:proof})  yields
\begin{equation*}
v_k\le \frac{2L_k\left(||\matrx U_0-\hat{\matrx U}||^2_F+||\matrx V_0-\hat{\matrx V}||^2_F\right)}{(k+1)^2}
\end{equation*}
By the Remark 3.2 in \cite{Beck:2009:FISTA} and the inequality (\ref{eq:upperL}), we also know that an upper bound of $L_k$ is $\tilde{L}$. Hence, 
\begin{equation*}
f(\matrx U_k,\matrx V_k)-f(\hat{\matrx U},\hat{\matrx V})\le \frac{2\tilde{L}\left(||\matrx U_0-\hat{\matrx U}||^2_F+||\matrx V_0-\hat{\matrx V}||^2_F\right)}{(k+1)^2}
\end{equation*}
In our algorithm, we set $L_k=\tilde{L}, \forall k$.
\end{proof}

\vspace{-2mm}
\begin{remark}
The loss function, $\ell(\matrx U, \matrx V)$, of an exponential distribution has Lipschitz continuous gradient within the range $\{||\matrx U||_{1,2}\le \delta_1, ||\matrx V^\top||_{1,2}\le \delta_2\}$ where $\delta_1, \delta_2$ are constant values in terms of $\lambda_1, \lambda_2$, respectively to guarantee the non-trivial step size $\frac{\lambda}{L}$. Otherwise, it may lead to a sub-optimal solution.
\end{remark}

\subsection{Asymptotic Analysis}\label{sec:Asy_Ana}
To facilitate the asymptotic analysis, we re-write the notation as follows: let 
\begin{equation*}
\vect \beta = [\mbox{vect}(\matrx U)^\top, \mbox{vect}(\matrx V)^\top]^\top, ~\mbox{ }~ \topidx{\matrx H}{i}= [\topidx{\vect h_{\tau+1}}{i}, \cdots, \topidx{\vect h_n}{i}]
\end{equation*}
and
\begin{equation*}
\topidx{\vect h_t}{i}=[\mbox{vect}(\matrx X_{i;t})^\top, \mbox{vect}(\matrx X_{i;t})^\top]^\top
\end{equation*}
where one block $\matrx X_{i;t}$ corresponds to $\matrx U$ and the other to $\matrx V$. Then, correspondingly, we have $\topidx{\eta}{i}_t=(\topidx{\vect h_t}{i})^\top\vect \beta$, and $f(\matrx U,\matrx V)$ can be re-written as $f(\vect \beta) = \ell(\vect \beta) + R(\vect \beta; \lambda_1,\lambda_2)$.

Solve Eq.(\ref{equ:dev_penality}) yields a solution to the penalized estimating equations:
\begin{equation}
\label{equ:pgee}
\sum_i (\topidx{\matrx D}{i})^\top (\topidx{\matrx \Sigma}{i})^{-1} \topidx{\vect s}{i} + \lambda \frac{\partial R(\vect \beta)}{\partial \vect \beta}=0
\end{equation}
assuming $\lambda_1 = \lambda_2=\lambda$ for notational convenience which will not change the property. Given our model definition (\ref{equ:AVR}), $\topidx{\matrx D}{i}=\topidx{\matrx A}{i}(\topidx{\matrx H}{i})^\top$. The first term in (\ref{equ:pgee}) is the estimating functions in GEE \cite{GEE:Liang:1986} whereas the second term corresponds to the regularizers. The asymptotic property of Eq.(\ref{equ:dev_penality}) can be naturally derived from the results in \cite{GEE:Liang:1986} which have proved that the estimating equations $L(\vect \beta) = \sum_i (\topidx{\matrx D}{i})^\top (\topidx{\matrx \Sigma}{i})^{-1} \topidx{\vect s}{i}$ of GEE gives a consistent estimator of $\vect \beta$. We extend the same argument to our formulation Eq.(\ref{equ:dev_penality}) in Theorem \ref{thm:asymptotics} under the following regularity conditions: 
$\topidx{\matrx H}{i}$ is bounded, and $\lim_{m \rightarrow \infty}(\sum_i \topidx{\matrx H}{i})/m=\topidx{\matrx H}{0}$, and $(\topidx{\matrx H}{i})^\top \topidx{\matrx H}{i}$ are not singular, and the following limit is also not singular
\begin{equation*}
\lim_{m \rightarrow \infty}(\sum_i (\topidx{\matrx H}{i})^\top \topidx{\matrx H}{i})/m;
\end{equation*}
Moreover, $L(\vect \beta)$ is twice continuously differentiable with respect to $\vect \beta$, and $\partial L/\partial \vect \beta$ is positive definite.

\begin{theorem}
\label{thm:asymptotics}
Assume that: (1) $\hat{\vect\alpha}$ is a consistent estimator given $\vect\beta$; (2) $\hat{\phi}$ is a consistent estimator given $\vect{\beta}$; and (3) the tuning parameter $\lambda_m=o(\sqrt{m})$.
Under the regularity conditions listed above, optimizing Eq.(\ref{equ:dev_penality}) yields an asymptotically consistent and normally distributed estimator $\hat{\vect \beta}$, that is: 
\begin{displaymath}\sqrt{m}(\hat{\vect \beta}-\vect \beta^*)\rightarrow_d N(0, \matrx \Sigma) \mbox{~~as~~} m\rightarrow\infty\end{displaymath} 
where $\vect \beta^*$ is the true model coefficients in a model of $E(y_t^{(i)}) = g^{-1}((\topidx{\vect h_t}{i})^\top\vect \beta)$ and $\Sigma$ is a positive definite variance-covariance matrix (see \cite{GEE:Liang:1986} for details of $\Sigma$). 
\end{theorem}

\begin{proof}
Multiplying $1/m$ to both sides of Eq.(\ref{equ:pgee}) yields
\begin{equation}
\label{equ:thm1_proof1}
\frac{1}{m}\sum_i (\topidx{\matrx D}{i})^\top (\topidx{\matrx \Sigma}{i})^{-1} \topidx{\vect s}{i} + \frac{\lambda_m}{m} \frac{\partial R(\vect \beta)}{\partial \vect \beta}=0.
\end{equation} 
It is known that solving $\frac{1}{m}\sum_i (\topidx{\matrx D}{i})^\top (\topidx{\matrx \Sigma}{i})^{-1} \topidx{\vect s}{i}=0$ yields an estimate of $\hat{\vect \beta}$ that is asymptotically consistent with $\vect \beta^*$:
\begin{equation*}
\sqrt{m}(\hat{\vect \beta}-\vect \beta^*)\rightarrow_d N(0, \matrx \Sigma) \mbox{~~as~~} m\rightarrow\infty ~~\cite{GEE:Liang:1986}.
\end{equation*}
Since our regularizer $R$ (based on the $\ell_{1,2}$ matrix norm) is Lipschitz continuous, its partial derivative $\partial R(\vect \beta)/\partial \vect \beta$ is bounded. The second term of Eq.(\ref{equ:thm1_proof1}) vanishes when $m \rightarrow \infty$, and thus the conclusion holds. 
\end{proof}

\vspace{-2mm}
Recall how $\hat{\vect \alpha}$ and $\hat{\phi}$ are estimated in the proposed method. Those estimates from the Pearson residuals are consistent. Thus, the estimate $\hat{\vect \beta}$ in the proposed method is asymptotically consistent and normally distributed according to Theorem 3.

\subsection{Exemplar Exponential Families with Lipschitz Condition}\label{sec:Eg_family}
The purposed algorithm is suitable to optimize any loss function that has Lipschitz continuous gradient. In this section, we discuss that three exemplar exponential families: Gaussian, Bernoulli, and Poisson, satisfy the Lipschitz condition. We specify how to compute the gradient of the loss function for these distributions. The gradients will instantiate (and replace) Eq.(\ref{equ:gradient}) used in our algorithm. 
\subsubsection{Gaussian Distribution}
If the outcome follows a Gaussian distribution, then the outcome $y$ is linearly regressive in terms of the covariates in the observations. The mean and the conditional covariance of $y$ with a working correlation structure $\matrx R(\vect \alpha)$ are calculated as:
\begin{align*}
E(y^{(i)}_t)&=\mu^{(i)}_t=tr\left(\matrx X^\top_{(i;t)}\matrx W\right),\\
cov(\vect y^{(i)})&=\matrx\Sigma^{(i)}=\matrx R(\vect \alpha),
\end{align*}
so the gradient $\nabla_{\matrx U}\ell_k$ in Eq.(\ref{equ:gradient}) at the $k$-th iteration can be computed as
\begin{equation*}
\label{equ:lr_gradient}
\nabla_{\matrx U}\ell_k = \mbox{reshape}\left(\sum^m_{i=1}\left(\matrx D^{(i)}\right)^\top\left(\matrx R(\vect{\alpha})\right)^{-1}\vect s^{(i)}_k\right),
\end{equation*}
where $\matrx D^{(i)}=\frac{\partial \vect{\mu}^{(i)}}{\partial \mathrm{vect}\left(\tilde{\matrx U}_k\right)} 
=\left[\mathrm{vect}\left(\matrx X_{(i;1)}\right), \dots, \mathrm{vect}\left(\matrx X_{(i;n)}\right) \right]^\top$, and $\vect s^{(i)}_k=\vect y^{(i)}-\left(\matrx D^{(i)}\right)^\top\mathrm{vect}(\tilde{\matrx U}_k)$. The gradient $\nabla_{\matrx V}\ell_k$ can be similarly computed. Hence, the gradient is linear in terms of $\vect \beta$, and thus Lipschitz continuous. 
\subsubsection{Bernoulli Distribution}
If the generalized variables $\mu$ follow a Bernoulli distribution and the outcomes are binary variables. The relationship between the outcome and covariates can be learned by a logistic regression which is a special case of the GLM with the Bernoulli assumption. Hence, the mean and the conditional covariance of $y$ with the working correlation structure $\matrx R(\vect \alpha)$ are formulated as
\begin{align}
E(y^{(i)}_t)&=\mu^{(i)}_t=\frac{\exp(\eta^{(i)}_t)}{1+\exp(\eta^{(i)}_t)} \label{equ:log_mu}\\
cov(\vect y^{(i)})=\matrx \Sigma^{(i)}&=\frac{\left(\matrx A^{(i)}\right)^{1/2}\matrx R(\vect \alpha)\left(\matrx A^{(i)}\right)^{1/2}}{\phi}\nonumber
\end{align}
where 
$\matrx A^{(i)}=\mathrm{diag}\left(\langle\vect\mu^{(i)},1-\vect\mu^{(i)}\rangle\right)$ \\    $=\mathrm{diag}\left(\frac{\exp(\eta^{(i)}_t)}{\left(1+\exp(\eta^{(i)}_t)\right)^2}\right)$ and 
$\eta^{(i)}_t=tr(\matrx X^\top_{(i;t)}\matrx W)$.

The gradient $\nabla_{\matrx U} \ell_k$ in Eq.(\ref{equ:gradient}) can be written as:
\begin{equation*}
\mbox{reshape}\left(\left(\matrx D^{(i)}\right)^\top (\matrx A^{(i)})^{-1/2}\matrx R(\vect \alpha)^{-1}(\matrx A^{(i)})^{-1/2}\topidx{\vect s}{i}_k\right)\label{equ:gradient_U}
\end{equation*}
where $
\matrx D^{(i)}=\frac{\partial \vect{\mu}^{(i)}}{\partial \vect \eta^{(i)}}\times\frac{\partial \vect{\eta}^{(i)}}{\partial \mathrm{vect}\left(\tilde{\matrx U}_k\right)}$\\
$=\matrx A^{(i)}\left[\mathrm{vect}\left(\matrx X_{(i;1)}\right), \dots, \mathrm{vect}\left(\matrx X_{(i;n)}\right)\right]^\top$, and 
$\vect s^{(i)}_k=\vect y^{(i)}-\vect \mu^{(i)}(\tilde{\matrx U}_k)$.
The gradient $\nabla_{\matrx V}\ell_k$ can be similarly computed. 
\subsubsection{Poisson Distribution}
If the generalized variables $\mu$ follow a Poisson distribution and the outcomes contain count values. The relationship of the outcome and covariates is learned by a Poisson regression.
The mean and the conditional covariance of $y$ with the working correlation structure $\matrx R(\vect \alpha)$ are formulated as
\begin{align*}
E(y^{(i)}_t)&=\mu^{(i)}_t=\exp(\eta^{(i)}_t) \label{equ:Poisson_mu}\\
cov(\vect y^{(i)})=\matrx \Sigma^{(i)}&=\frac{\left(\matrx A^{(i)}\right)^{1/2}\matrx R(\vect \alpha)\left(\matrx A^{(i)}\right)^{1/2}}{\phi}\nonumber
\end{align*}
where $\matrx A^{(i)}=\mathrm{diag}\left((\vect\mu^{(i)})'\right)$
              $=\mathrm{diag}\left(\exp(\eta^{(i)}_t)\right)$.
The gradient $\nabla_{\matrx U} \ell_k$ can be computed using the general formula Eq.(\ref{equ:gradient}). 
The loss function of Poisson regression does not have globally Lipschitz continuous gradient. But the regularized loss function is equivalent to requiring the constraints, $||\matrx U||_{1,2}\le \delta_1$ and $||\matrx V^\top||_{1,2}\le \delta_2$ \cite{osborne2000lasso} for appropriate values of $\delta_1$ and $\delta_2$ that are determined according to $\lambda_1$ and $\lambda_2$. The loss function of Poisson regression does have Lipschitz continuous gradient within the confined region.

\section{Empirical Evaluation}
\label{sec:evaluation}
We validated the proposed approach by comparing it to several most relevant and recent methods. Three GLM-based  \cite{GLM2002olsson} methods: GEE \cite{GEE:Liang:1986}, GLMM \cite{Laird:1982:GLMM,McCulloch:2001:GLMM}, and  RE-EM tree\footnote{An R package is available in the Comprehensive R Archive Network (CRAN)} \cite{REEMtree} were compared. The recent graphical Granger modeling\footnote{downloaded from the author's website http://www-bcf.usc.edu/$\sim$liu32/code.html} \cite{Lozano2009} and a support vector machine based method called CSVM were also used. RE-EM tree and graphical Granger modeling could only be applied to regression problems (linearly regressive data from Gaussian distributions), and CSVM was only suitable to classification tasks (logistically regressive data from Bernoulli distributions). We named our approach by LGL (longitudinal group lasso). The normalized mean squared error (nMSE), which is the MSE divided by the variance of $y$ \cite{multi:Zhang:2010,MTL:2012:Gong}, was used to measure regression performance.  The area under the ROC curve (AUC) \cite{roc2006brown} was used to measure classification performance.

\subsection{Synthetic Data}
We generated a data matrix $\matrx X \in \mathrm{\mathbb{R}}^{d\times Tm}$ from the normal distribution $N(0,16)$, where $d=200$, $T=30$, and $m=400$. All training examples $\matrx X_{(i;t)}$($i=1,\cdots,m$, $\forall t=\tau+1,\cdots,T$) and $\tau=4$ were formed from the matrix $\matrx X$. Then, $\matrx U$ and $\matrx V$ were generated from the normal distribution $N(0,49)$. We set the rows corresponding to features from 1 to 150 in $\matrx U$ to zero and the columns 2 and 5 of $\matrx V$ to zero, and computed $\matrx W = \matrx U + \matrx V$. The  residuals $\vect s^{(i)}$ of every subject were generated from a multivariate normal distribution of different variances, $N(0,1^2),N(0,2^2),N(0,3^2)$. The covariance matrix of the residual followed different working correlation structures $\matrx R(\alpha)$ with the parameter $\alpha=0.64$. We generated 9 sets of regression residuals by choosing different combinations of the variances and the working correlation structures. Finally, the outcome variables $\vect y^{(i)}$ were computed as
\begin{equation*}
\vect y^{(i)}=\left[\mathrm{vect}\left(\matrx X_{(i;\tau+1)}\right),\dots,\mathrm{vect}\left(\matrx X_{(i;n)}\right)\right]^\top \mathrm{vect}(\matrx U+\matrx V)+\vect s^{(i)}.
\end{equation*}
The above procedure produced regression data. Using the same data $\matrx X$, the outcome $y^{(i)}_t$ of a classification problem was generated from the Bernoulli Distribution with $\mathrm{B}(1,\mu^{(i)}_t)$ where we used Eq.(\ref{equ:log_mu}) with the regression $\vect y^{(i)}$ to obtain $\vect \mu^{(i)}$. We hence obtained totally 18 synthesized data with 9 datasets for each distribution. We used the 25 early records of each subject to compose the training data and the rest 5 records to form test data. 

\begin{figure}[tbp]
\centering
	\includegraphics[width=.37\textwidth]{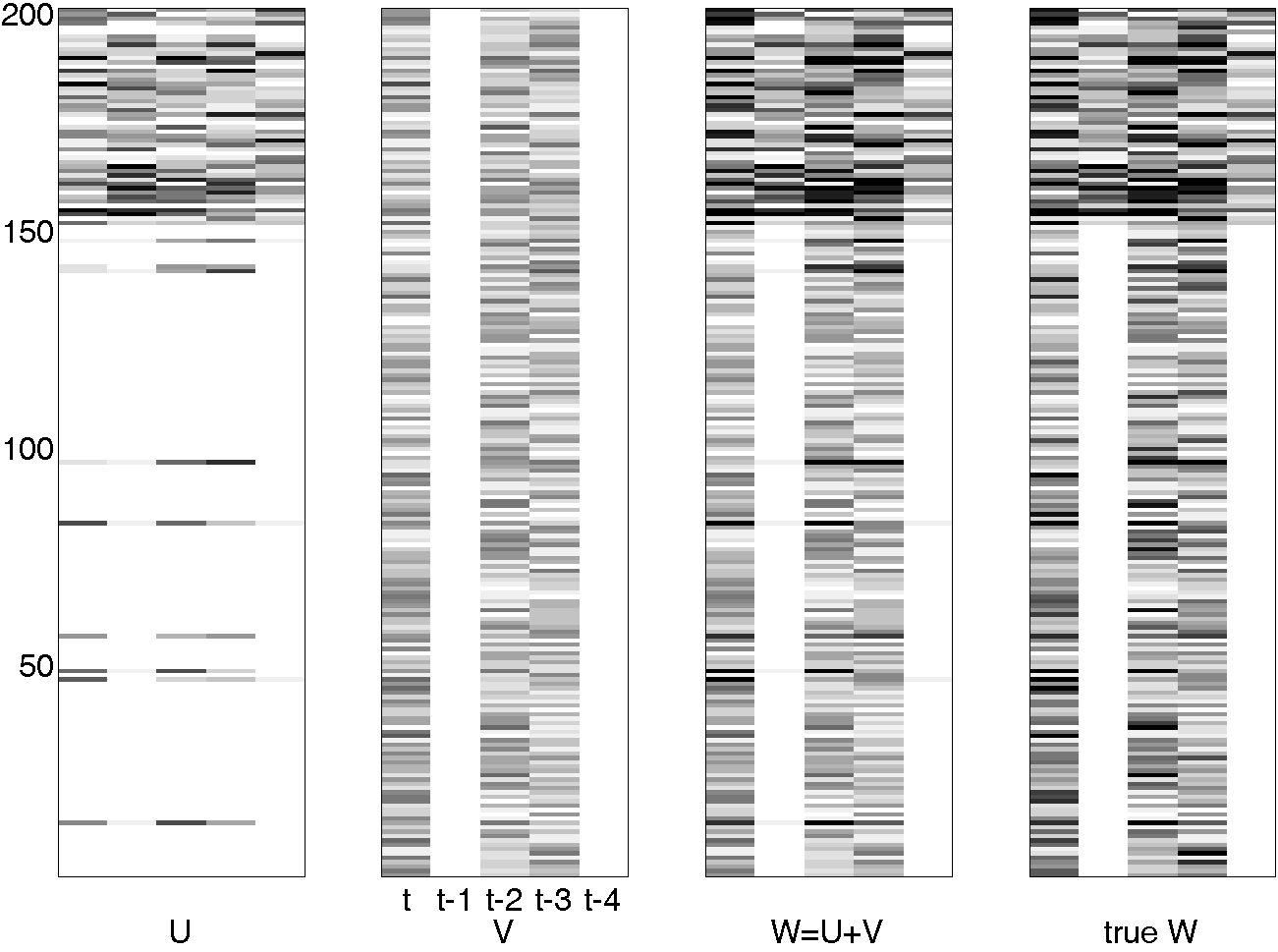} \vspace{-0.1in}
	\caption{The model constructed by our approach LGL on a synthetic dataset.}
	\label{fig:pcolorW} \vspace{-0.1in}
\end{figure}
\begin{figure}[tbp]
\centering
	\includegraphics[width=.35\textwidth]{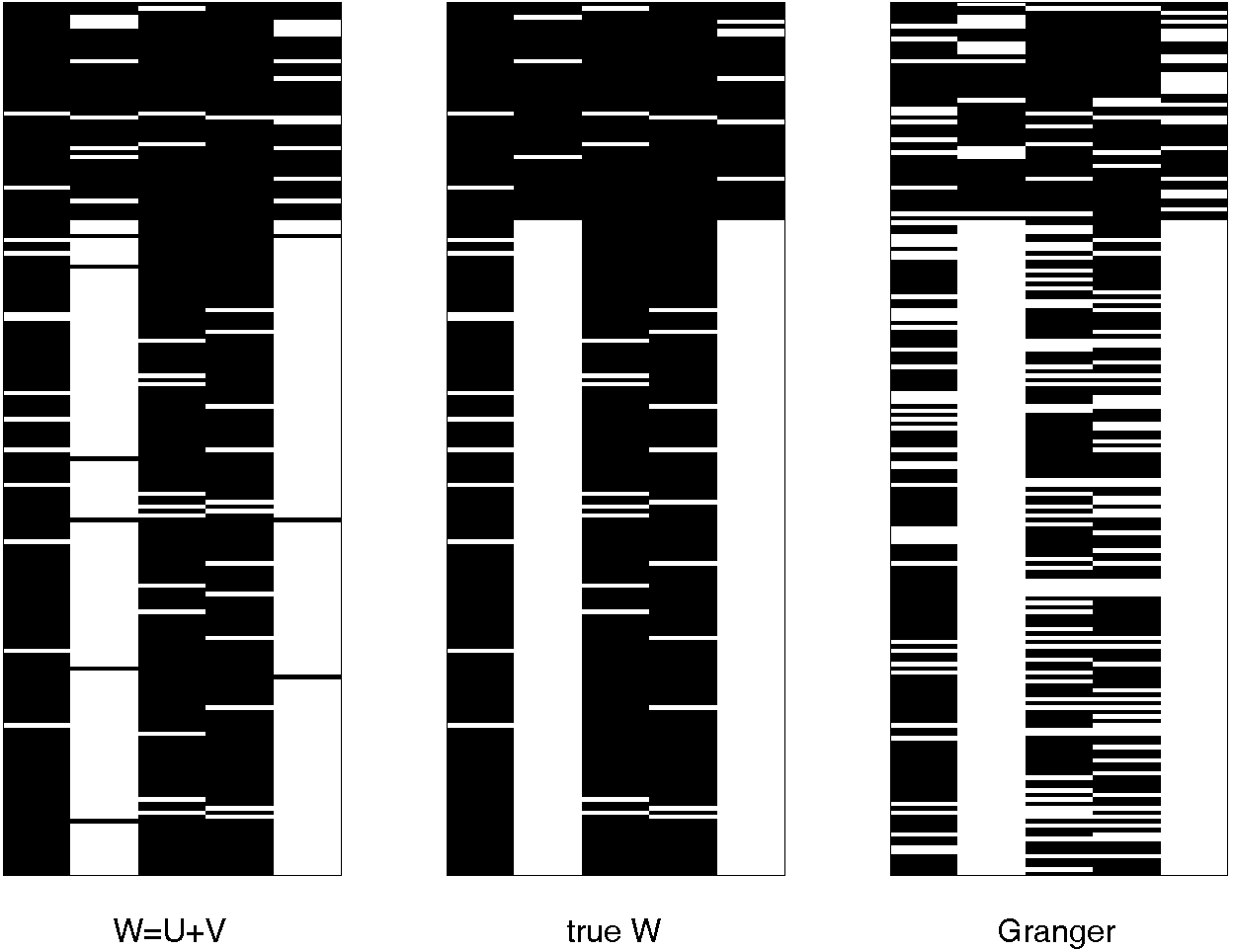} \vspace{-0.1in}
	\caption{Comparison between the constructed models by LGL and Granger.}
	\label{fig:pcolorW_granger} \vspace{-0.1in}
\end{figure}

Table \ref{tab:MSEsyn} shows the results where we can see that LGL outperformed all other methods on all the simulated datasets. The proposed method with correct correlation assumptions always performed the best. The graphical Granger modeling performed reasonably well but lacked of consideration of temporal correlation in the consecutive records. When the simulated noise increased, the performance of all methods had dropped as expected. We further demonstrate the selected features and temporal contingency. Figure \ref{fig:pcolorW} shows the constructed $\matrx U, \matrx V$, and $\matrx W$ by the LGL on the regression data with the AR(1) covariance structure and $N(0,3^2)$ residual where darker colors indicate larger values (and white means 0).  Most of the features from 150 to 200 were selected in $\matrx U$ and the correct columns (i.e., ${1,3,4}$) were selected in $\matrx V$. We compared our approach with the Granger model that also learned $\matrx W$ in Figure \ref{fig:pcolorW_granger}. Obviously, the Granger model excluded too many variables in the model. These results demonstrate the capability of LGL in terms of simultaneously capturing the important features and lagged effects.

\subsection{Real-world Data}
We tested our approach on two real-world datasets: the college alcohol use dataset; and the national longitudinal survey of youth (NLSY) dataset\footnote{http://www.bls.gov/nls/nlsy97.htm}. All comparison methods were used except GLMM due to its prohibitive computational costs. The college alcohol use dataset consisted of data from 504 college students on 52 variables in a period of continuous 30 days. The 52 variables measured each subject on daily stress, moods, emotion and substance use behavior. One of the variables measured the number of night-time drinks, which was our outcome variable, forming a regression problem. We also predicted the binge drinking behavior which is defined as having 5 or more night-time drinks, which formed a classification problem.
The NLSY dataset consisted of 11 yearly data for 3,376 subjects on 27 variables. The outcome variable measured the number of days that a subject had binge drinking in past 30 days, forming a regression problem. The other 26 variables measured features, such as smoking, drug use, family support and education.

For the college alcohol use data, we experimented with using the last $t=3,5,8,10 $ days of records as test data, and the rest for training. We found $\tau = 3$ was feasible. Larger $\tau$ would not change the results because the extra time points would be excluded by our model. However, it practically would cut down the sample size of each subject. The parameters $\lambda_1$ and $\lambda_2$ in our approach and any tuning parameters in other methods were tuned in a three-fold cross validation within the training data. Table \ref{tab:MSErec} shows the results where our approach LGL outperformed other methods in most settings. Among the four different correlation assumptions, LGL with AR(1) obtained the best performance on three of the four settings. The results also confirmed that modeling the correlation among repeated observations improved prediction performance \cite{GEE:Liang:1986}. We also observed that for instance, 16 out of 51 variables were selected when we used the last $5$ days to test binge drinking prediction. Features related to exited mood, under stress and interacting with friends during night time were the risk factors for binge drinking. The past 3 days were all included in the model, showing there was ``lagged" effects in alcohol use. The effect of past days was reduced with prolonged time lag.
\begin{table*}[t]
\begin{center}
\begin{scriptsize}
\caption{Comparison of different algorithms on synthetic data: (top) regression; (bottom) classification.}
\label{tab:MSEsyn}
	\begin{tabular}{@{}r@{ }|@{ }c@{ }|@{ }c@{ }|@{ }c@{}c@{ }c@{}c@{}|@{ }c@{}c@{ }c@{}c@{}|@{ }c@{ }@{ }c@{ }c@{}}
	\hline\hline
	\parbox[t]{2mm}{\multirow{11}{*}{\rotatebox[origin=c]{90}{Regression}}}&&&\multicolumn{4}{c|}{LGL}&\multicolumn{4}{c|}{GEE}&& \multicolumn{1}{|@{ }c@{ }}{} & \multicolumn{1}{|@{ }c@{}}{}\\
	& Structures& $e$ & AR(1) & exchangeable & Tri-diag & ind & AR(1) & exchangeable & Tri-diag & ind & GLMM & \multicolumn{1}{|@{ }c@{ }}{RE-EM tree} & \multicolumn{1}{|@{ }c@{}}{Granger} \\
	\cline{2-14}
		&&$N(0,1^2)$& \textbf{0.0018} & 0.0020 & 0.0019 & 0.0020 & 0.6613 & 0.6615 & 0.6614 & 0.6617 & 0.6657 & \multicolumn{1}{|@{ }c@{ }}{0.9873} & \multicolumn{1}{|@{ }c@{}}{0.0664} \\
	& AR(1)  &$N(0,2^2)$& \textbf{0.0025} & 0.0026 & 0.0028 & 0.0039 & 0.7223 & 0.7236 & 0.7224 & 0.7242 & 0.7323 & \multicolumn{1}{|@{ }c@{ }}{0.9998} & \multicolumn{1}{|@{ }c@{}}{0.0667} \\
		&&$N(0,3^2)$& \textbf{0.0032} & 0.0034 & 0.0036 & 0.0038 & 0.7191 & 0.7185 & 0.7182 & 0.7192 & 0.7179 & \multicolumn{1}{|@{ }c@{ }}{0.9924} & \multicolumn{1}{|@{ }c@{}}{0.0676} \\
	\cline{2-14}
			 &&$N(0,1^2)$& 0.0018 & 0.0016 & \textbf{0.0015} & 0.0022 & 0.6872 & 0.6875 & 0.6872 & 0.6873 & 0.6914 & \multicolumn{1}{|@{ }c@{ }}{0.9977} & \multicolumn{1}{|@{ }c@{}}{0.0656} \\
	& exchangeable&$N(0,2^2)$& 0.0024 & \textbf{0.0023} & 0.0024 & 0.0025 & 0.6927 & 0.6930 & 0.6927 & 0.6930 & 0.6931 & \multicolumn{1}{|@{ }c@{ }}{0.9982} & \multicolumn{1}{|@{ }c@{}}{0.0691} \\
			 &&$N(0,3^2)$& 0.0027 & \textbf{0.0026} & 0.0028 & 0.0032 & 0.7204 & 0.7204 & 0.7204 & 0.7205 & 0.7204 & \multicolumn{1}{|@{ }c@{ }}{0.9797} & \multicolumn{1}{|@{ }c@{}}{0.0635} \\
	\cline{2-14}
			&&$N(0,1^2)$ & 0.0021 & 0.0021 & \textbf{0.0021} & 0.0022 & 0.7514 & 0.7514 & 0.7514 & 0.7514 & 0.7515 & \multicolumn{1}{|@{ }c@{ }}{0.9925} & \multicolumn{1}{|@{ }c@{}}{0.0665} \\
	& Tri-diag&$N(0,2^2)$ & 0.0018 & 0.0023 & \textbf{0.0013} & 0.0026 & 0.6790 & 0.6792 & 0.6791 & 0.6793 & 0.6840 & \multicolumn{1}{|@{ }c@{ }}{0.9991} & \multicolumn{1}{|@{ }c@{}}{0.0680} \\
			&&$N(0,3^2)$ & 0.0033 & 0.0035 & \textbf{0.0031} & 0.0041 & 0.7226 & 0.7235 & 0.7226 & 0.7226 & 0.7222 & \multicolumn{1}{|@{ }c@{ }}{0.9998} & \multicolumn{1}{|@{ }c@{}}{0.0660} \\
	\hline\hline
	\parbox[t]{2mm}{\multirow{11}{*}{\rotatebox[origin=c]{90}{Classification}}}&&&\multicolumn{4}{c|}{LGL}&\multicolumn{4}{c|}{GEE} & & &\\
	& Structures& $e$ & AR(1) & exchangeable & Tri-diag & ind & AR(1) & exchangeable & Tri-diag & ind & CSVM & & \\
	\cline{2-12}
			&&$N(0,1^2)$ & \textbf{96.490\%} & 96.485\% & 96.485\% & 96.417\% & 77.691\% & 77.700\% & 77.699\% & 77.715\% & 76.644\% & & \\
	& AR(1)   &$N(0,2^2)$ & \textbf{96.442\%} & 96.431\% & 96.432\% & 96.653\% & 74.682\% & 74.727\% & 74.682\% & 74.731\% & 75.249\% & & \\
			&&$N(0,3^2)$ & \textbf{95.921\%} & 95.917\% & 95.917\% & 95.805\% & 77.704\% & 77.746\% & 77.708\% & 77.754\% & 77.547\% & & \\
	\cline{2-12}
			 &&$N(0,1^2)$ & 95.913\% & \textbf{95.937\%} & 95.912\% & 95.883\% & 76.115\% & 75.812\% & 76.114\% & 75.923\% & 75.232\% & & \\
	& exchangeable&$N(0,2^2)$ & 95.139\% & \textbf{95.161\%} & 95.147\% & 95.150\% & 70.290\% & 70.231\% & 70.275\% & 70.206\% & 71.687\% & & \\
			 &&$N(0,3^2)$ & 94.127\% & 94.091\% & \textbf{94.135\%} & 93.470\% & 73.839\% & 73.782\% & 73.831\% & 73.776\% & 73.894\% & & \\
	\cline{2-12}
			&&$N(0,1^2)$ & 95.976\% & 95.941\% & \textbf{95.978\%} & 95.889\% & 77.628\% & 77.634\% & 77.625\% & 77.617\% & 76.778\% & & \\
	& Tri-diag&$N(0,2^2)$ & 95.231\% & 95.231\% & \textbf{95.245\%} & 94.395\% & 72.132\% & 72.060\% & 72.126\% & 72.054\% & 71.615\% & & \\
			&&$N(0,3^2)$ & 95.092\% & 95.087\% & \textbf{95.094\%} & 94.231\% & 77.755\% & 77.533\% & 77.748\% & 77.637\% & 77.572\% & & \\
	\cline{1-12}\morecmidrules\cmidrule{1-12}
	\end{tabular}
\end{scriptsize}
\end{center}
\end{table*}
\begin{table*}[htbp] \vspace{-0.2in}
\begin{center}
\begin{scriptsize}
\caption{Comparison of different algorithms on the college alcohol use dataset: (top) predicting the number of night-time drinks (regression); (bottom) predicting the occurrence of binge drinking (classification).}
\label{tab:MSErec}
	\begin{tabular}{@{}r@{ }|@{ }c|c@{ }c@{ }c@{ }c@{ }|c@{ }c@{ }c@{ }c@{ }|cc@{ }}
 \hline \hline
  \parbox[t]{2mm}{\multirow{6}{*}{\rotatebox[origin=c]{90}{Regression}}}&& \multicolumn{4}{c|}{LGL} & \multicolumn{4}{c|}{GEE} & & \multicolumn{1}{|c@{ }}{} \\
 & \# observations & AR(1) & exchangeable & tri-diag & ind &AR & exchangeable & tri-diag & ind & RE-EM tree & \multicolumn{1}{|c@{ }}{Granger} \\
 \cline{2-12}
 & 3 & \textbf{0.933513} & 0.933863 & 0.935120 & 0.961841 & 1.064792 & 1.073358 & 1.063948 & 1.065760 & 1.115627 & \multicolumn{1}{|c@{ }}{1.369948} \\
 & 5 & 0.951999 & 0.954740 & \textbf{0.951953} & 0.976299 & 1.051219 & 1.067303 & 1.049305 & 1.072745 & 1.005753 & \multicolumn{1}{|c@{ }}{1.420547} \\
 & 8 & \textbf{0.759935} & 0.760450 & 0.760136 & 0.762205 & 0.787731 & 0.793329 & 0.787497 & 0.794089 & 0.759968 & \multicolumn{1}{|c@{ }}{0.909706} \\
 & 10 & \textbf{0.769303} & 0.769492 & 0.769428 & 0.774937 & 0.812622 & 0.818834 & 0.812011 & 0.806301 & 0.774797 & \multicolumn{1}{|c@{ }}{0.940940}\\
 \hline \hline
  \parbox[t]{2mm}{\multirow{6}{*}{\rotatebox[origin=c]{90}{Classification}}}&& \multicolumn{4}{c|}{LGL} & \multicolumn{4}{c|}{GEE} & & \\
 & \# observations & AR(1) & exchangeable & tri-diag & ind &AR & exchangeable & tri-diag & ind & CSVM & \\
 \cline{2-11}
 & 3 & 79.737\% & 75.677\% & 79.772\% & 78.579\% & 78.401\% & 74.145\% & 78.650\% & 77.831\% & \textbf{80.698\%}& \\
 & 5 & \textbf{83.290\%} & 77.237\% & 83.070\% & 82.323\% & 80.371\% & 78.363\% & 80.646\% & 80.438\% & 83.187\% & \\
 & 8 & \textbf{88.570\%} & 87.331\% & 87.936\% & 87.787\% & 85.999\% & 86.330\% & 85.714\% & 86.014\% & 88.017\% & \\
 & 10 & \textbf{89.484\%} & 87.574\% & 88.853\% & 88.578\% & 85.979\% & 86.622\% & 85.721\% & 85.783\% & 89.041\% & \\
 \cline{1-11}\morecmidrules\cmidrule{1-11}
\end{tabular}
\end{scriptsize}
\end{center}
\end{table*}

\begin{figure}[t]
\centering
	\includegraphics[width=.43\textwidth]{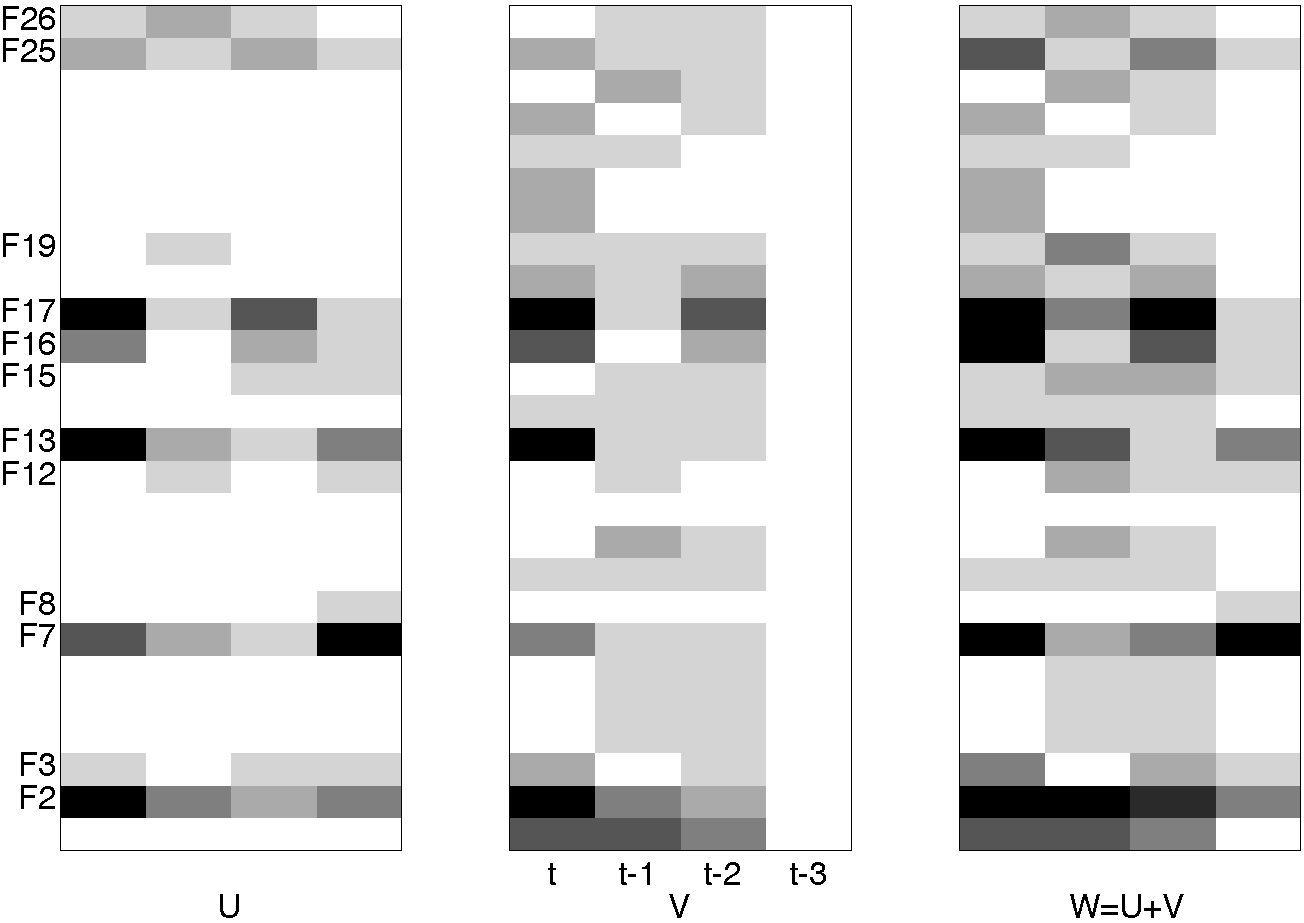} \vspace{-0.1in}
	\caption{The model constructed by our approach on the NLSY dataset.} \vspace{-0.1in}
	\label{fig:featureNLSY}
\end{figure}
For the NLSY dataset, we experimented respectively with using the last one, two and three years from each subject for test and the rest in training. We also considered $\tau = 3$, which means we used 3 year lagged data to predict the current year's behavior. All tuning parameters were tuned using a within-training two-fold cross validation.  The  results are reported in Table \ref{tab:MSEnlsy}. For any assumption of the working correlation structure, LGL had comparative performance with RE-EM tree and consistently outperformed GEE in all of the three experiments. LGL with tri-diagonal correlation performed the best on this dataset. The results here again show that taking care of the correlation among repeated observations improves the performance (given we see that LGL with the independent correlation assumption had the worst performance among all LGL variants).  

The gray map of $\matrx U$, $\matrx V$ and $\matrx W$ constructed by LGL is shown in Figure \ref{fig:featureNLSY} to illustrate an example for the tri-diagonal working correlation assumption. Out of the 26 features, 12 were selected by LGL and we list them below. \\
\setlength{\itemsep}{0cm}%
\setlength{\parskip}{0cm}%
\setlength{\parsep}{0cm}
\textbf{F2:} \# days of smoking a cigarette in the past 30 days \\
\textbf{F3:} Received a training certificate or vocational license \\
\textbf{F7:} The grade began during the academic year \\
\textbf{F8:} \# months that respondent did not attend school during  the academic year \\
\textbf{F12:} The college degree working toward or attained\\
\textbf{F13:} The highest grade completed as of the survey year \\
\textbf{F15:} The highest grade attended as of the survey day \\
\textbf{F16:} The highest grade completed as of the survey day \\
\textbf{F17:} \# days of using marijuana in the past 30 days \\
\textbf{F19:} \# times of using some drug or other substance right  before school or during school or work hours \\
\textbf{F25:} As the victim of a violent crime in the survey year \\
\textbf{F26:} Divorced parents.

\begin{table*}[htbp] 
\begin{center}
\begin{scriptsize}
\caption{Comparison of different algorithms on the NLSY dataset in terms of test nMSE values.}
\label{tab:MSEnlsy}
	\begin{tabular}{@{ }c|c@{ }c@{ }c@{ }c@{ }|c@{ }c@{ }c@{ }c@{ }|c|c@{ }}
 \hline \hline
  & \multicolumn{4}{c|}{LGL} & \multicolumn{4}{c|}{GEE} & & \\
 \# observations & AR(1) & exchangeable & tri-diag & ind &AR & exchangeable & tri-diag & ind & RE-EM tree & Granger \\
 \hline
 1 & 0.906552 & 0.908932 & 0.904760 & 0.909446 & 0.911543 & 0.918691 & 0.911885 & 0.914043 & \textbf{0.904260} & 1.370135\\
 2 & 0.888608 & 0.891761 & \textbf{0.887294} & 0.891051 & 0.898132 & 0.904225 & 0.897920& 0.898320 & 0.888822 & 1.363714\\
 3 & 0.885448 & 0.885814 & \textbf{0.883617} & 0.887579 & 0.892963 & 0.895863 & 0.892633 & 0.890937 & 0.883958 & 1.360430\\
 \hline \hline
\end{tabular}
\end{scriptsize}
\end{center} \vspace{-0.3in}
\end{table*}

\noindent This list shows that a subject's smoking, drug use, education background and family support influenced his or her drinking behavior. Figure \ref{fig:featureNLSY} demonstrates that the data in the third prior year might be obsolete to predict this year's behavior as LGL only selected the past two years for use in the model as seen in the plot of $\matrx V$.

\section{Discussion}
\label{sec:discussion}
We have proposed a new learning formulation for longitudinal analytics. Unlike existing methods, the proposed approach can simultaneously determine the temporal contingency and the influential features in predicting an outcome over time. The model parameter matrix is computed by the summation of two component matrices: one matrix reflects the selection among covariates; and the other characterizes the dependency along the temporal line. Moreover, our approach simultaneously models the sample correlations in the longitudinal data while constructing a predictive model. The related optimization problem can be efficiently solved by a new accelerated gradient descent algorithm. Convergence analysis shows that the algorithm can find the global optimal solution for the model with a quadratic convergence rate. An asymptotic analysis shows that the solution of our formulation is a consistent estimate of the model parameters. Hence, the proposed approach solves an underdeveloped problem - jointly learning the relevant features and determining how current outcome relies on past observations.  Empirical studies on both synthetic and real-world problems demonstrate the superior performance of the proposed approach over the state of the art.

\section*{Acknowledgments}
This work was supported by NSF grants IIS-1320586, DBI-1356655 and NIH grant R01DA037349. Jinbo Bi was also
supported by NSF grants IIS-1407205 and IIS-1447711.

\bibliographystyle{abbrv}
\bibliography{biblio_Tingyang,biblio_medicine_javon,biblio_jinbo}
\end{document}